\documentclass[twocolumn,amsthm]{autart}

\usepackage{amsmath,amssymb}
\usepackage{mathtools}
\usepackage{tikz}
\usetikzlibrary{arrows.meta,positioning,fit,backgrounds,calc,decorations.pathreplacing,calligraphy}
\usepackage{bm}
\usepackage{enumitem}
\usepackage{booktabs}

\newcommand{\R}{\mathbb{R}}

\newcommand{\Linf}{L^\infty}
\newcommand{\Lip}{\mathrm{Lip}}

\newcommand{\Lipk}{\mathrm{Lip}_K}

\newcommand{\eps}{\varepsilon}

\theoremstyle{plain}
\newtheorem{theorem}{Theorem}
\newtheorem{lemma}{Lemma}

\theoremstyle{definition}
\newtheorem{definition}{Definition}
\newtheorem{assumption}{Assumption}
\newtheorem{remark}{Remark}

\allowdisplaybreaks
\begin{document}

\begin{frontmatter}

\title{\bfseries Approximate Feedback Linearization for a
Nonlinear Hyperbolic PDE Class --- Part II:
Neural Operator}
\label{title:main}

\author{Miroslav Krstic}
\ead{mkrstic@ucsd.edu}
\thanks{Department of Mechanical and Aerospace Engineering, University of California San Diego, La Jolla, CA 92093-0411.}
\thanks{The principal AI aid in developing the paper was Claude.}

\begin{abstract}
Volterra series feedback linearizes a class of nonlinear hyperbolic PDEs but produces a controller that, even after truncation, demands solving a tower of plant-specific kernel PDEs and evaluating nested integrals. We prove the truncated controller is jointly Lipschitz in plant and state, and learn it as a single neural operator from plant nonlinearity and state to boundary control. Once trained, no kernel is ever solved again, for any plant in the trained class. The closed loop is practically stable in class-$\mathcal{KL}$ form, with a residual ball scaling linearly with training accuracy.
\end{abstract}

\end{frontmatter}

\section{Introduction}
\label{sec:intro}

This is the second of two companion papers on approximate feedback linearization of a class of nonlinear hyperbolic PDEs. Part~I~\cite{krstic2026feedbacklinearizationhyperbolicpdes-trunc} replaces the exact infinite Volterra linearizer of~\cite{krstic2026feedbacklinearizationhyperbolicpdes} with a finite truncation and establishes its closed-loop stability in $L^\infty$. Part~II, the present paper, replaces the truncated controller with a learned neural-operator surrogate that eliminates the online kernel solve and the nested Volterra integration, and re-establishes the same closed-loop guarantees under two simultaneous approximations.

Volterra series feedback turns a nonlinear hyperbolic PDE into a transport equation, exactly. The price is a controller defined by an infinite series of nested integrals over kernels living on simplices of growing dimension --- a feedback no implementation can evaluate and no offline solver can precompute. Truncation at some finite order $N$ is forced upon us, and the truncated controller is what gets implemented. The question is what truncation does to the closed loop.

Part~I~\cite{krstic2026feedbacklinearizationhyperbolicpdes-trunc} answers it. With the sup-norm machinery developed there, the truncated feedback delivers, on a sup-norm ball whose radius is set by the contraction radius of the inverse backstepping transformation, four properties: forward invariance, a practical-stability bound holding for all time, finite-time attractivity to a residual ball determined by the truncation tail $\eps_N(r_0)$, and a class-$\mathcal{KL}$ asymptotic estimate. Stabilization of~\cite{krstic2026feedbacklinearizationhyperbolicpdes} is recovered in the high-order limit.

But the truncated controller, while finite, is still computationally expensive: each evaluation of $\mathcal{U}_N(\mathbf{f}_N, u(\cdot,t))$ requires an $N$-fold nested integration against $N - 1$ separately-precomputed kernels $k_2, \ldots, k_N$, with each $k_n$ obtained by solving a linear PDE on an $n$-simplex. Worse, this kernel computation is plant-specific: change the nonlinearity $F$, and every $k_n$ must be solved for again from scratch. The kernels are not the controller; they are an intermediate object the controller is forced to traverse.

The proper object of study is the operator $\mathcal{U}_N(\mathbf{f}_N, u)$ that maps the plant nonlinearity, with Volterra coefficients $\mathbf{f}_N = (f_2, \ldots, f_N)$, together with the state, to the boundary input. This is what gets implemented in the loop, and it is what we propose to learn. A neural operator trained on a class of admissible plants emulates $\mathcal{U}_N$ on that whole class at once. The kernel PDEs are solved offline only as part of generating training data; once the network is trained, no kernel is ever solved again --- not when the plant nonlinearity is changed within the trained class, not when the state is updated in real time, not at all. A single forward pass of fixed cost replaces the nested-integral computation, for any plant in the class.

This note is the analysis of the closed loop under such a neural surrogate. Once $\mathcal{U}_N$ is replaced by a sup-norm-accurate approximation $\widehat{\mathcal{U}}_N$ with error $\eps$ on the joint $(\mathbf{f}_N, u)$ space, the boundary residual driving the target transport equation acquires a new contribution: $b(t) = -[\widehat{\mathcal{U}}_N - \mathcal{U}_N] + [\mathcal{U}_N - \mathcal{U}]$, the first piece of size $\eps$ from the network, the second of size $\eps_N$ from the truncation. The four stability properties of~\cite{krstic2026feedbacklinearizationhyperbolicpdes-trunc} are re-established under this combined residual, and the asymptotic envelope takes the conventional form $\|u(\cdot, t)\|_\infty \le \beta(\|u_0\|_\infty, t) + \gamma(\eps)$, with $\beta$ of class $\mathcal{KL}$ and $\gamma$ a linear class-$\mathcal{K}$ function. As $\eps \to 0$, the truncated paper's class-$\mathcal{KL}$ behavior is recovered.

The contributions of this note over~\cite{krstic2026feedbacklinearizationhyperbolicpdes-trunc} are:
(i) a Lipschitz-continuity result for the truncated Volterra feedback operator $\mathcal{U}_N$ as a joint function of the plant kernels and the state, which serves as the basis for the universal-approximation guarantee in (ii);
(ii) existence of a neural-operator approximation $\widehat{\mathcal{U}}_N$ achieving any prescribed sup-norm accuracy on bounded sets of plants and states;
(iii) the practical-stability theorem under neural-operator feedback, with explicit linear-in-$\eps$ residual envelope and a sharpened initial-condition bound that tightens linearly with $\eps$.

The truncated paper made the controller implementable in principle. This note makes it implementable in practice, once and for all per plant class.

The paper is organized as follows. Sections~\ref{sec:operator}--\ref{sec:proof-thm1} establish the joint Lipschitz continuity of $\mathcal{U}_N$ in coefficients and state (Theorem~\ref{thm:main}); Section~\ref{sec:NO} produces the neural-operator surrogate (Theorem~\ref{thm:NO}); Sections~\ref{sec:closedloop}--\ref{sec:proof-stab} close the loop and prove practical stability under the surrogate feedback (Theorem~\ref{thm:stab}); Section~\ref{sec:numerics} demonstrates the framework at $N=3$; and Section~\ref{sec:conclusion} concludes. Appendix~\ref{app:numerical-details} gathers the numerical details.

\section{Operator Definition}
\label{sec:operator}

We consider the hyperbolic plant
\begin{eqnarray}
u_t(x,t) &=& u_x(x,t) + F[u(\cdot,t)](x), \nonumber \\
&& u(0,t)=0, \quad u(1,t)=U(t),
\label{eq:plant}
\end{eqnarray}
with initial condition $u(\cdot, 0) = u_0$ and plant nonlinearity
\begin{eqnarray}
F[u](x,t) &=& \sum_{n=2}^\infty \int_{T_n(x)} f_n(x,\xi_1,\ldots,\xi_n) \nonumber \\
&& \times \prod_{i=1}^n u(\xi_i,t)\, d\xi_n\cdots d\xi_1,
\label{eq:Fdef}
\end{eqnarray}
where $T_n(x) := \{(\xi_1,\ldots,\xi_n) : 0 \le \xi_n \le \cdots \le \xi_1 \le x\}$ is the standard simplex and $\Delta_n := \{(x,\xi_1,\ldots,\xi_n) : 0 \le \xi_n \le \cdots \le \xi_1 \le x \le 1\}$. The control input is $U(t)$, applied at the boundary $x = 1$.

\begin{assumption}[Bound on Volterra coefficients]
\label{ass:fnbound}
There exist constants $D_f, \rho_f > 0$ such that
\begin{equation}
\|f_n\|_{\Linf} \;\le\; D_f\,\frac{n!}{\rho_f^{\,n-1}},
\quad n \ge 2.
\label{eq:fnbound}
\end{equation}
\end{assumption}

The order-$N$ truncated boundary controller is the operator
\begin{eqnarray}
\mathcal{U}_N(\mathbf{f}_N, u)
&=& \mathcal{U}_{N,k}(\mathbf{k}_N, u) \nonumber \\
&:=& \sum_{n=2}^{N}\, \int_{T_n(1)} k_n(1,\xi_1,\ldots,\xi_n) \nonumber \\
&& \times \prod_{i=1}^{n} u(\xi_i)\, d\xi_n\cdots d\xi_1,
\label{eq:KNdef}
\end{eqnarray}
acting jointly on the plant Volterra coefficients $\mathbf{f}_N = (f_2, \ldots, f_N)$ and the state $u$, with $\mathbf{k}_N = (k_2, \ldots, k_N)$ the kernels generated from $\mathbf{f}_N$ via \eqref{eq:k2def}--\eqref{eq:kndef} below. The exact (untruncated) feedback operator is its limit $\mathcal{U}(\mathbf{f}_N, u) := \lim_{N\to\infty} \mathcal{U}_N(\mathbf{f}_N, u)$, which is the boundary feedback law of the hyperbolic feedback-linearization construction~\cite{krstic2026feedbacklinearizationhyperbolicpdes}. As established in~\cite{krstic2026feedbacklinearizationhyperbolicpdes-trunc}, on every set $\{u \in \Linf(0,1) : \|u\|_\infty \le r\}$ with $r \in [0, 1/C_K)$, the truncation residual is bounded by
\begin{eqnarray}
|\mathcal{U}(\mathbf{f}_N, u) - \mathcal{U}_N(\mathbf{f}_N, u)| &\le& \eps_N(\|u\|_\infty), \nonumber \\
\eps_N(r) &:=& \frac{D_K\, e^{\Upsilon_K}\, (C_K r)^{N+1}}{C_K\,(1 - C_K r)},
\label{eq:epsN-def}
\end{eqnarray}
where $D_K, C_K, \Upsilon_K > 0$ are constants determined by the plant Volterra coefficients $\mathbf{f}_N$. The kernels $\{k_n\}_{n\ge 2}$ in \eqref{eq:KNdef} are generated from $\mathbf{f}_N$ by the characteristic recursion of~\cite{krstic2026feedbacklinearizationhyperbolicpdes},
\begin{equation}
k_2(x,\xi_1,\xi_2) = -\int_0^{\xi_2} f_2(x-\xi_2+s,\, \xi_1-\xi_2+s,\, s)\, ds,
\label{eq:k2def}
\end{equation}
\begin{eqnarray}
k_n(x,\xi_1,\ldots,\xi_n) &=& -\int_0^{\xi_n}\Bigl[f_n \nonumber \\
&& -\, \sum_{m=2}^{n-1} B_n^m[k_{n-m+1},f_m]\Bigr] \nonumber \\
&& \times\, \bigl(x-\xi_n+s,\, \xi_1-\xi_n+s, \nonumber \\
&& \quad \ldots,\, \xi_{n-1}-\xi_n+s,\, s\bigr)\, ds, \quad n\ge 3,
\label{eq:kndef}
\end{eqnarray}
where the bilinear coupling operator $B_n^m$ is given, with the convention $\xi_0 := x$, by
\begin{eqnarray}
&& B_n^m[k,f](x,\xi_1,\ldots,\xi_n) \nonumber \\
&& = \sum_{j=1}^{n-m+1}\int_{\xi_j}^{\xi_{j-1}}
\bigl[D_j^{n-m+1,m} k\bigr](x,\xi_1,\ldots, \nonumber \\
&& \quad\ \ \xi_{j-1},s,\xi_j,\ldots,\xi_{n-m}) \nonumber \\
&& \quad\ \ \times f(s,\xi_{n-m+1},\ldots,\xi_n)\,ds,
\label{eq:Bdef}
\end{eqnarray}
with the permutation-summation operator
\begin{eqnarray}
&& \bigl[D_j^{p,m}g\bigr](x,\zeta_1,\ldots,\zeta_{p+m-1}) \nonumber \\
&& = \!\!\!\sum_{\substack{(\gamma_1,\ldots,\gamma_{p+m-1-j}) \\ \in\, P_{p-j}(\zeta_{j+1},\ldots,\zeta_{p+m-1})}}\!\!\!
g(x,\zeta_1,\ldots,\zeta_{j-1}, \nonumber \\
&& \qquad\qquad\qquad \zeta_j,\gamma_1,\ldots,\gamma_{p+m-1-j}),
\label{eq:Ddef}
\end{eqnarray}
where $P_{p-j}(\zeta_{j+1},\ldots,\zeta_{p+m-1})$ is the set of ordered $(p+m-1-j)$-tuples whose first $p-j$ entries are any $p-j$ elements of $\{\zeta_{j+1},\ldots,\zeta_{p+m-1}\}$ in their original order, with the remaining $m-1$ entries being the leftover elements of the same set, also in their original order. The number of such tuples is $\binom{p+m-1-j}{p-j}$.

We define the admissible plant-coefficient class and the admissible state class.

\begin{definition}[Admissible classes]
\label{def:classes}
For integers $N\ge 2$ and constants $B_f, L_f, R, L_u > 0$, set
\begin{eqnarray}
\lefteqn{\mathcal{F}_n(B_f,L_f) := \bigl\{ f \in C^{0,1}([0,1]^{n+1}) :} \nonumber \\
&& \|f\|_{\Linf} \le B_f,\ \Lip(f) \le L_f \bigr\}, \ 2\le n\le N,
\label{eq:Fnclass}
\end{eqnarray}
\begin{eqnarray}
\mathcal{F}_N^\star(B_f,L_f) &:=& \prod_{n=2}^N \mathcal{F}_n(B_f,L_f), \nonumber \\
\|\mathbf{f}_N\|_\star &:=& \max_{2\le n\le N} \|f_n\|_{\Linf},
\label{eq:Fclass}
\end{eqnarray}
\begin{eqnarray}
\mathcal{V}(R,L_u) &:=&
\bigl\{ u \in C^{0,1}([0,1]) : \nonumber \\
&& \|u\|_{\Linf(0,1)} \le R,\ \Lip(u) \le L_u \bigr\}.
\label{eq:Vclass}
\end{eqnarray}
The Lipschitz constant $\Lip$ refers throughout to the Euclidean Lipschitz seminorm on the relevant simplex or interval.
\end{definition}

\section{Theorem: Operator is Lipschitz}
\label{sec:theorem}

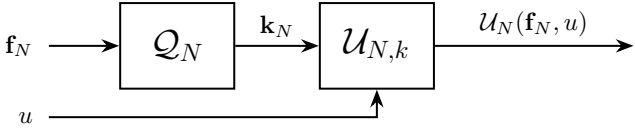
\begin{figure}[t]
\centering
\resizebox{\columnwidth}{!}{%
\begin{tikzpicture}[
    >={Stealth[length=2.5mm]},
    box/.style={draw, thick, minimum width=16mm, minimum height=12mm, font=\large},
    every node/.style={font=\normalsize},
]
\coordinate (fNin) at (0, 0.5);
\coordinate (uin)  at (0, -0.5);
\node[box] (QN) at (1.8, 0.5) {$\mathcal{Q}_N$};
\node[box] (UN) at (4.6, 0.5) {$\mathcal{U}_{N,k}$};
\coordinate (outR) at (8.2, 0.5);
\node[left=1mm of fNin] {$\mathbf{f}_N$};
\node[left=1mm of uin]  {$u$};
\draw[->, thick] (fNin) -- (QN.west);
\draw[->, thick] (QN.east) -- node[above] {$\mathbf{k}_N$} (UN.west);
\draw[->, thick] (uin) -- ++(4.6, 0) -- (UN.south);
\draw[->, thick] (UN.east) -- node[above] {$\mathcal{U}_N(\mathbf{f}_N, u)$} (outR);
\end{tikzpicture}%
}
\caption{Cascade of operators $\mathcal{Q}_N$ and $\mathcal{U}_{N,k}$ generates feedback control input $\mathcal{U}_N(\mathbf{f}_N, u)$.}
\label{fig:cascade}
\end{figure}

Two operations of different character connect $(\mathbf{f}_N, u)$ to $\mathcal{U}_N(\mathbf{f}_N, u)$ in Figure~\ref{fig:cascade}. The following theorem certifies that, jointly, they are Lipschitz.

\begin{theorem}[Joint Lipschitz property]
\label{thm:main}
Fix $N\ge 2$ and $B_f, L_f, R, L_u > 0$. The joint controller operator $\mathcal{U}_N$ defined in \eqref{eq:KNdef}, restricted to $\mathcal{U}_N \colon \mathcal{F}_N^\star(B_f,L_f) \times \mathcal{V}(R,L_u) \to \R$, satisfies
\begin{eqnarray}
\bigl|\mathcal{U}_N(\mathbf{f}_N,u) - \mathcal{U}_N(\tilde{\mathbf{f}}_N,\tilde u)\bigr| &\le& M_N(R)\,\|\mathbf{f}_N-\tilde{\mathbf{f}}_N\|_\star \nonumber \\
&& +\, M'_N(R)\,\|u-\tilde u\|_{\Linf(0,1)}
\label{eq:joint-Lip}
\end{eqnarray}
for all $\mathbf{f}_N,\tilde{\mathbf{f}}_N\in\mathcal{F}_N^\star(B_f,L_f)$ and all $u,\tilde u\in\mathcal{V}(R,L_u)$, where the operator Lipschitz constants are
\begin{equation}
M_N(R) := \sum_{n=2}^N \frac{R^n}{n!}\,\lambda_n, \quad M'_N(R) := \sum_{n=2}^N \frac{n R^{n-1}}{n!}\,\kappa_n,
\label{eq:M-def}
\end{equation}
with the kernel Lipschitz-in-$\mathbf{f}_N$ constants $\lambda_2 := 1$ and, for $3\le n\le N$,
\begin{equation}
\lambda_n := 1 + \sum_{m=2}^{n-1} \beta_{n,m}\bigl(B_f\,\lambda_{n-m+1} + \kappa_{n-m+1}\bigr),
\label{eq:lambda-def}
\end{equation}
the kernel sup-norm bounds $\kappa_2 := B_f$ and, for $3\le n\le N$,
\begin{equation}
\kappa_n := B_f + B_f\sum_{m=2}^{n-1} \beta_{n,m}\,\kappa_{n-m+1},
\label{eq:kappa-def}
\end{equation}
and the combinatorial constants
\begin{eqnarray}
\beta_{n,m} &:=& (n-m+1) \nonumber \\
&& \times\, \max_{1\le j\le n-m+1}\binom{n-j}{n-m+1-j} \nonumber \\
&\le& n\,2^n, \quad 2\le m\le n\le N.
\label{eq:beta-def}
\end{eqnarray}
In particular, $\mathcal{U}_N$ is locally Lipschitz on bounded sets in $\mathcal{F}_N^\star(B_f,L_f)\times\mathcal{V}(R,L_u)$.
\end{theorem}

The proof proceeds in two stages: first, that the kernels $k_n$ themselves are bounded and Lipschitz-continuous in $\mathbf{f}_N$ in sup-norm with constants $\kappa_n$ and $\lambda_n$ (Section~\ref{sec:knLip}); second, that the operator $\mathcal{U}_N(\mathbf{f}_N, u)$, as a multilinear functional, is Lipschitz in $(k_2, \ldots, k_N, u)$ (Section~\ref{sec:UNlip}). The two stages are then composed (Section~\ref{sec:proof}).

\begin{remark}
\label{rem:beta-bound}
The closed-form bound $\beta_{n,m}\le n\,2^n$ in \eqref{eq:beta-def} follows from $\binom{n-j}{n-m+1-j}\le 2^{n-j}\le 2^{n-1}$ and $n-m+1\le n$.
\end{remark}

\begin{remark}
\label{rem:kappa-growth}
The constants $\kappa_n, \lambda_n$, and hence $M_N(R), M'_N(R)$, grow super-exponentially in $N$, so the bound \eqref{eq:joint-Lip} is far from tight. It is not meant to be: the seminorms $L_f, L_u$ are absent from it entirely, serving only to render the input sets $\mathcal{F}_n(B_f,L_f)$ and $\mathcal{V}(R,L_u)$ compact in $\Linf$ (Arzel\`a--Ascoli). The bound is qualitative --- it certifies the continuity of $\mathcal{U}_N$ that the neural-operator universal approximation theorem of Section~\ref{sec:NO} requires.
\end{remark}

\section{Proof of Theorem~\ref{thm:main}}
\label{sec:proof-thm1}

\subsection{Sup-norm Bound on $B_n^m$}
\label{sec:Bbound}

All upper bounds on the kernels $k_n$ trace back to bounds on the bilinear blocks $B_n^m$ feeding the recursion. We start there.

\begin{lemma}[Bound on $B_n^m$]
\label{lem:Bbound}
For all $k\in\Linf(T_{n-m+1}(1))$, $f\in\Linf(T_m(1))$, and $2\le m\le n$,
\begin{eqnarray}
&& \|B_n^m[k,f]\|_{\Linf(\Delta_n)} \nonumber \\
&& \le\, \beta_{n,m}\,\|k\|_{\Linf(T_{n-m+1}(1))}\,\|f\|_{\Linf(T_m(1))},
\label{eq:Bbound}
\end{eqnarray}
with $\beta_{n,m}$ as in \eqref{eq:beta-def}.
\end{lemma}

\begin{proof}
For each $j\in\{1,\ldots,n-m+1\}$, the operator $D_j^{n-m+1,m}$ in \eqref{eq:Ddef} is, by definition, a sum over a set of cardinality $\binom{p+m-1-j}{p-j}$ with $p = n-m+1$, which equals $\binom{n-j}{n-m+1-j}$, of evaluations of $k$ at points of $T_{n-m+1}(1)$. Hence pointwise,
\begin{equation}
\bigl|\bigl[D_j^{n-m+1,m}k\bigr](\cdot)\bigr|
\le \binom{n-j}{n-m+1-j}\,\|k\|_{\Linf}.
\label{eq:Dbnd}
\end{equation}
The integral over $[\xi_j,\xi_{j-1}]\subseteq[0,1]$ contributes a factor at most $1$, and the outer summation is over $n-m+1$ values of $j$. Therefore, pointwise on $\Delta_n$,
\begin{eqnarray}
&& |B_n^m[k,f](\cdot)| \nonumber \\
&& \le (n-m+1)\,\max_{1\le j\le n-m+1}\binom{n-j}{n-m+1-j} \nonumber \\
&& \quad \times\, \|k\|_{\Linf}\,\|f\|_{\Linf} \nonumber \\
&& =\, \beta_{n,m}\,\|k\|_{\Linf}\,\|f\|_{\Linf}.
\label{eq:Bbnd-pw}
\end{eqnarray}
\end{proof}

\subsection{Lipschitz Dependence of $k_n$ on Plant Coefficients}
\label{sec:knLip}

We assemble the kernels generated from $\mathbf{f}_N$ via \eqref{eq:k2def}--\eqref{eq:kndef} into a tuple $\mathbf{k}_N = (k_2, \ldots, k_N)$, and write the kernel-generation recursion as the operator $\mathbf{k}_N = \mathcal{Q}_N(\mathbf{f}_N)$.
Lemma~\ref{lem:knBound} below establishes a sup-norm bound on $\mathcal{Q}_N$ on bounded sets, and Lemma~\ref{lem:knLip} establishes its Lipschitz continuity.

\begin{lemma}[Sup-norm bound for $k_n$]
\label{lem:knBound}
For every $\mathbf{f}_N\in\mathcal{F}_N^\star(B_f,L_f)$, the kernels $k_n$ defined by \eqref{eq:k2def}--\eqref{eq:kndef} satisfy
\begin{equation}
\|k_n\|_{\Linf(\Delta_n)} \le \kappa_n,
\qquad 2\le n\le N,
\label{eq:knBound}
\end{equation}
with $\kappa_n$ as in \eqref{eq:kappa-def}.
\end{lemma}

\begin{proof}
For $n=2$, \eqref{eq:k2def} gives, pointwise,
\begin{equation}
|k_2(x,\xi_1,\xi_2)| \le \int_0^{\xi_2} \|f_2\|_{\Linf}\,ds \le B_f,
\label{eq:k2-pw}
\end{equation}
since $\xi_2\in[0,1]$ and $\|f_2\|_{\Linf}\le B_f$. Hence $\|k_2\|_{\Linf}\le B_f = \kappa_2$.

For the inductive step, suppose \eqref{eq:knBound} holds for $2\le n'\le n-1$. From \eqref{eq:kndef}, pointwise on $\Delta_n$, $|k_n(x,\xi_1,\ldots,\xi_n)| \le \int_0^{\xi_n} [\|f_n\|_{\Linf} + \sum_{m=2}^{n-1} \|B_n^m[k_{n-m+1},f_m]\|_{\Linf}]\,ds$. Applying Lemma~\ref{lem:Bbound} and the inductive hypothesis $\|k_{n-m+1}\|_{\Linf}\le\kappa_{n-m+1}$, together with $\|f_m\|_{\Linf}\le B_f$,
\begin{eqnarray}
&& |k_n(x,\xi_1,\ldots,\xi_n)| \nonumber \\
&& \le \int_0^{\xi_n}\Bigl[B_f + B_f\sum_{m=2}^{n-1} \beta_{n,m}\,\kappa_{n-m+1}\Bigr]\,ds \nonumber \\
&& \le B_f + B_f\sum_{m=2}^{n-1} \beta_{n,m}\,\kappa_{n-m+1}
\;=\; \kappa_n,
\label{eq:kn-pw-2}
\end{eqnarray}
using $\xi_n\in[0,1]$ and \eqref{eq:kappa-def}.
\end{proof}

Lemma~\ref{lem:knBound} controls the size of each $k_n$. We now control how $k_n$ moves when the plant coefficients $\mathbf{f}_N$ are perturbed.

\begin{lemma}[Lipschitz dependence of $k_n$]
\label{lem:knLip}
For every $\mathbf{f}_N,\tilde{\mathbf{f}}_N\in\mathcal{F}_N^\star(B_f,L_f)$, with $k_n,\tilde k_n$ the corresponding kernels defined by \eqref{eq:k2def}--\eqref{eq:kndef},
\begin{equation}
\|k_n - \tilde k_n\|_{\Linf(\Delta_n)}
\le \lambda_n\,\|\mathbf{f}_N-\tilde{\mathbf{f}}_N\|_\star,
\label{eq:knLip}
\end{equation}
for $2\le n\le N$, with $\lambda_n$ as in \eqref{eq:lambda-def}.
\end{lemma}

\begin{proof}
For $n=2$, subtract \eqref{eq:k2def} for $\tilde f_2$ from that for $f_2$, and bound pointwise:
\begin{eqnarray}
&& |k_2(x,\xi_1,\xi_2) - \tilde k_2(x,\xi_1,\xi_2)| \nonumber \\
&& \le \int_0^{\xi_2} \|f_2 - \tilde f_2\|_{\Linf}\,ds \nonumber \\
&& \le \|f_2 - \tilde f_2\|_{\Linf}
\;\le\; \|\mathbf{f}_N-\tilde{\mathbf{f}}_N\|_\star,
\label{eq:k2diff-pw}
\end{eqnarray}
so \eqref{eq:knLip} holds with $\lambda_2 = 1$.

For the inductive step, suppose \eqref{eq:knLip} holds for $2\le n'\le n-1$. The bilinearity of $B_n^m$ in its two arguments gives the splitting
\begin{eqnarray}
&& B_n^m[k_{n-m+1},f_m] - B_n^m[\tilde k_{n-m+1},\tilde f_m] \nonumber \\
&& = B_n^m[k_{n-m+1}-\tilde k_{n-m+1},\,f_m] \nonumber \\
&& \quad +\, B_n^m[\tilde k_{n-m+1},\,f_m - \tilde f_m].
\label{eq:B-bilinear-split}
\end{eqnarray}
Applying Lemma~\ref{lem:Bbound} to each term, and using $\|k_{n-m+1}-\tilde k_{n-m+1}\|_{\Linf}\le \lambda_{n-m+1}\|\mathbf{f}_N-\tilde{\mathbf{f}}_N\|_\star$ (inductive hypothesis), $\|f_m\|_{\Linf}\le B_f$, $\|\tilde k_{n-m+1}\|_{\Linf}\le\kappa_{n-m+1}$ (Lemma~\ref{lem:knBound}), and $\|f_m-\tilde f_m\|_{\Linf}\le\|\mathbf{f}_N-\tilde{\mathbf{f}}_N\|_\star$,
\begin{eqnarray}
&& \|B_n^m[k_{n-m+1},f_m] - B_n^m[\tilde k_{n-m+1},\tilde f_m]\|_{\Linf} \nonumber \\
&& \le \beta_{n,m}\bigl(B_f\,\lambda_{n-m+1} + \kappa_{n-m+1}\bigr) \nonumber \\
&& \quad \times\, \|\mathbf{f}_N-\tilde{\mathbf{f}}_N\|_\star.
\label{eq:Bdiff-bnd}
\end{eqnarray}
Subtracting the $\tilde{\mathbf{f}}_N$ version of \eqref{eq:kndef} from the $\mathbf{f}_N$ version and applying these bounds termwise,
\begin{eqnarray}
|k_n - \tilde k_n| &\le& \int_0^{\xi_n}\Bigl[\|\mathbf{f}_N-\tilde{\mathbf{f}}_N\|_\star + \sum_{m=2}^{n-1} \beta_{n,m} \nonumber \\
&& \times\bigl(B_f\,\lambda_{n-m+1} + \kappa_{n-m+1}\bigr) \nonumber \\
&& \times\, \|\mathbf{f}_N-\tilde{\mathbf{f}}_N\|_\star\Bigr]\,ds \nonumber \\
&\le& \lambda_n\,\|\mathbf{f}_N-\tilde{\mathbf{f}}_N\|_\star,
\label{eq:kndiff-pw}
\end{eqnarray}
using $\xi_n\le 1$ and \eqref{eq:lambda-def}.
\end{proof}

\subsection{Lipschitz Dependence of $\mathcal{U}_N(\mathbf{f}_N, u)$ on $(\mathbf{k}_N, u)$}
\label{sec:UNlip}

The operator $\mathcal{U}_{N,k}$ in \eqref{eq:KNdef} is a sum of $n$-fold integrals $J_n$ over $n = 2, \ldots, N$. We bound each $J_n$ separately, then sum.

\begin{lemma}[Multilinear bound]
\label{lem:single-order}
Fix $n\in\{2,\ldots,N\}$ and define
\begin{equation}
J_n[k,u] := \int_{T_n(1)} k(\xi_1,\ldots,\xi_n)\,\prod_{i=1}^n u(\xi_i)\,d\xi_n\cdots d\xi_1,
\label{eq:Jdef}
\end{equation}
for $k\in\Linf(T_n(1))$ and $u\in\Linf(0,1)$. Then for $\|k\|_{\Linf},\|\tilde k\|_{\Linf}\le K$ and $\|u\|_{\Linf},\|\tilde u\|_{\Linf}\le R$,
\begin{eqnarray}
&& |J_n[k,u] - J_n[\tilde k,\tilde u]| \nonumber \\
&& \le \frac{R^n}{n!}\,\|k-\tilde k\|_{\Linf}
+ \frac{K\,n R^{n-1}}{n!}\,\|u-\tilde u\|_{\Linf}.
\label{eq:JLip}
\end{eqnarray}
\end{lemma}

\begin{proof}
The simplex volume is $|T_n(1)| = 1/n!$. Write
\begin{eqnarray}
&& J_n[k,u] - J_n[\tilde k,\tilde u] \nonumber \\
&& = J_n[k-\tilde k,u] + J_n[\tilde k,\,u^{\otimes n} - \tilde u^{\otimes n}],
\label{eq:Jdiff-split}
\end{eqnarray}
where $u^{\otimes n}(\xi) := \prod_i u(\xi_i)$ and similarly for $\tilde u$. The first term satisfies $|J_n[k-\tilde k,u]| \le \|k-\tilde k\|_{\Linf}\,R^n / n!$. For the second, the multilinear telescoping identity
\begin{eqnarray}
\prod_{i=1}^n u(\xi_i) - \prod_{i=1}^n \tilde u(\xi_i) &=& \sum_{j=1}^n \Bigl[\prod_{i<j} u(\xi_i)\Bigr]\bigl(u(\xi_j) - \tilde u(\xi_j)\bigr) \nonumber \\
&& \times\, \Bigl[\prod_{i>j}\tilde u(\xi_i)\Bigr]
\label{eq:telescope}
\end{eqnarray}
yields, on $T_n(1)$ with $\|u\|_{\Linf},\|\tilde u\|_{\Linf}\le R$,
\begin{equation}
\bigl|u^{\otimes n}(\xi) - \tilde u^{\otimes n}(\xi)\bigr|
\le n R^{n-1}\,\|u-\tilde u\|_{\Linf},
\label{eq:tensor-diff}
\end{equation}
hence $|J_n[\tilde k,\,u^{\otimes n} - \tilde u^{\otimes n}]| \le K\,n R^{n-1}\,\|u-\tilde u\|_{\Linf} / n!$. Adding gives \eqref{eq:JLip}.
\end{proof}

\subsection{Completing the proof of Theorem~\ref{thm:main}}
\label{sec:proof}

Fix $\mathbf{f}_N,\tilde{\mathbf{f}}_N\in\mathcal{F}_N^\star(B_f,L_f)$ and $u,\tilde u\in\mathcal{V}(R,L_u)$, and let $k_n,\tilde k_n$ be the kernels generated from $\mathbf{f}_N$ and $\tilde{\mathbf{f}}_N$ via \eqref{eq:k2def}--\eqref{eq:kndef}. By Lemma~\ref{lem:knBound}, $\|k_n\|_{\Linf(\Delta_n)},\|\tilde k_n\|_{\Linf(\Delta_n)}\le\kappa_n$ for $2\le n\le N$. By Lemma~\ref{lem:knLip}, $\|k_n-\tilde k_n\|_{\Linf(\Delta_n)}\le\lambda_n\,\|\mathbf{f}_N-\tilde{\mathbf{f}}_N\|_\star$. Since the boundary traces $k_n(1,\cdot)$ are restrictions of $k_n$ to the slice $\{x = 1\}$ of $\Delta_n$, the same bounds hold with $\Linf(T_n(1))$ in place of $\Linf(\Delta_n)$.

Applying Lemma~\ref{lem:single-order} to each summand of \eqref{eq:KNdef} with $K=\kappa_n$,
\begin{eqnarray}
|\mathcal{U}_N(\mathbf{f}_N,u) - \mathcal{U}_N(\tilde{\mathbf{f}}_N,\tilde u)| &\le& \sum_{n=2}^N\Bigl[\frac{R^n}{n!}\,\|k_n-\tilde k_n\|_{\Linf} \nonumber \\
&& +\, \frac{\kappa_n\,n R^{n-1}}{n!}\,\|u-\tilde u\|_{\Linf}\Bigr].
\label{eq:UN-diff-bnd}
\end{eqnarray}
Substituting $\|k_n-\tilde k_n\|_{\Linf}\le\lambda_n\,\|\mathbf{f}_N-\tilde{\mathbf{f}}_N\|_\star$ and using \eqref{eq:M-def} yields \eqref{eq:joint-Lip}. Local Lipschitz continuity on bounded sets is immediate.

\section{Neural Operator Approximation}
\label{sec:NO}

\begin{figure*}[t]
\centering
\resizebox{0.81\linewidth}{!}{%
\begin{tikzpicture}[
    >={Stealth[length=2.5mm]},
    innerbox/.style={draw, thick, minimum width=34mm, minimum height=14mm, font=\normalsize, align=center},
    every node/.style={font=\normalsize},
]
\node[innerbox] (PDE) at (3.0, 0) {kernel PDE\\ solver};
\node[innerbox] (NI)  at (7.5, 0) {nested integration\\ and summation};
\draw[->, thick] (PDE.east) -- (NI.west);
\begin{pgfonlayer}{background}
  \node[draw, dashed, thick, rounded corners,
        inner xsep=3mm, inner ysep=6mm,
        fit=(PDE)(NI)] (NN) {};
\end{pgfonlayer}
\coordinate (inAnchor) at ($(NN.west) + (-1.6, 0)$);
\node[left=0pt of inAnchor] (in) {$(\mathbf{f}_N, u)$};
\coordinate (outAnchor) at ($(NN.east) + (3.4, 0)$);
\draw[->, thick] (in.east) -- (NN.west);
\draw[->, thick] (NN.east) -- node[above] {$\mathcal{U}_N(\mathbf{f}_N, u)$} (outAnchor);
\node[below=2mm of NN.south, font=\large] {used to train NN $\widehat{\mathcal{U}}_N$};
\end{tikzpicture}%
}
\caption{Neural operator $\widehat{\mathcal{U}}_N$ subsumes the cascade of Figure~\ref{fig:cascade} --- kernel PDE solution and nested integration --- into a single trained mapping from $(\mathbf{f}_N, u)$ to $\mathcal{U}_N(\mathbf{f}_N, u)$.}
\label{fig:NO-cascade}
\end{figure*}
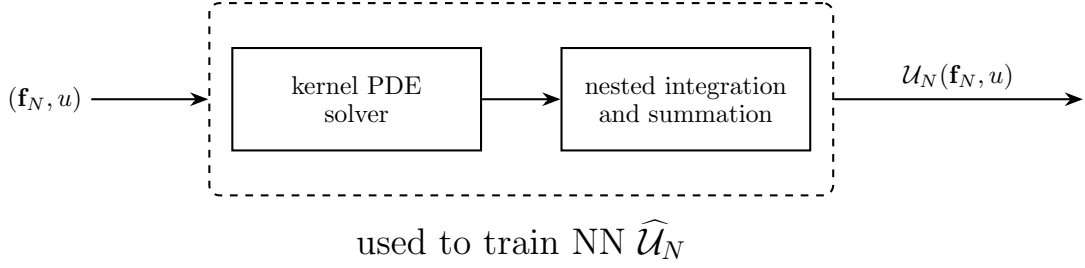

The neural-operator universal approximation theorem~\cite{Lu2021Universal}, applied to operators continuous on compact sets, ensures the existence of a neural-operator surrogate $\widehat{\mathcal{U}}_N$ for $\mathcal{U}_N$ within any prescribed accuracy (Figure~\ref{fig:NO-cascade}).

\begin{theorem}[Neural-operator surrogate]
\label{thm:NO}
For every $\eps > 0$, there exists a neural operator $\widehat{\mathcal{U}}_N : \mathcal{F}_N^\star(B_f,L_f) \times \mathcal{V}(R,L_u) \to \R$ such that
\begin{eqnarray}
&& \bigl|\mathcal{U}_N(\mathbf{f}_N, u) - \widehat{\mathcal{U}}_N(\mathbf{f}_N, u)\bigr| \le \eps, \nonumber \\
&& \forall\,(\mathbf{f}_N, u) \in \mathcal{F}_N^\star(B_f,L_f) \times \mathcal{V}(R,L_u).
\label{eq:NObnd}
\end{eqnarray}
\end{theorem}

\begin{proof}
By Theorem~\ref{thm:main}, $\mathcal{U}_N$ is Lipschitz on $\mathcal{F}_N^\star(B_f,L_f)\times\mathcal{V}(R,L_u)$, hence continuous; the input set is compact in $\Linf$ by Arzel\`a--Ascoli. The result then follows from the universal approximation theorem of~\cite{Lu2021Universal}.
\end{proof}

\section{Closed Loop and Target System}
\label{sec:closedloop}

We close the loop on the hyperbolic plant \eqref{eq:plant}--\eqref{eq:Fdef} with the neural-operator feedback $U(t) = \widehat{\mathcal{U}}_N(\mathbf{f}_N, u(\cdot, t))$. The exact backstepping transformation $w = \mathcal{T}[u] = u - K[u]$ from~\cite{krstic2026feedbacklinearizationhyperbolicpdes} maps the closed-loop plant into a transport target system. The interior equation $w_t = w_x$
holds. At the boundary $x = 1$, writing $u = u(\cdot,t)$, the feedback gives $w(1,t) = u(1,t) - \mathcal{U}(\mathbf{f}_N, u) = \widehat{\mathcal{U}}_N(\mathbf{f}_N, u) - \mathcal{U}(\mathbf{f}_N, u)$. Adding and subtracting $\mathcal{U}_N(\mathbf{f}_N, u)$ decomposes this into the neural-operator error and the truncation residual:
\begin{eqnarray}
b(t) = w(1,t) &=& \bigl[\widehat{\mathcal{U}}_N(\mathbf{f}_N, u) - \mathcal{U}_N(\mathbf{f}_N, u)\bigr] \nonumber \\
&& +\, \bigl[\mathcal{U}_N(\mathbf{f}_N, u) - \mathcal{U}(\mathbf{f}_N, u)\bigr].
\label{eq:bdef-NO}
\end{eqnarray}
Thus $w$ satisfies the linear transport equation $w_t = w_x$ on $(0,1)$ with boundary residual $b(t)$ and initial condition $w_0 = \mathcal{T}[u_0]$. The first bracket in \eqref{eq:bdef-NO} is the neural-operator approximation error from Theorem~\ref{thm:NO}; the second bracket is the truncation tail bounded by $\eps_N(\|u(\cdot,t)\|_\infty)$, established in~\cite{krstic2026feedbacklinearizationhyperbolicpdes-trunc}.

\section{Practical Stability under Neural-Operator Feedback}
\label{sec:stability}

The closed-loop analysis follows the structure of~\cite{krstic2026feedbacklinearizationhyperbolicpdes-trunc}, with the boundary residual $b(t)$ now bounded by the sum of the neural-operator error $\eps$ and the truncation tail $\eps_N(r_0)$. We use, without restatement, the following directly importable results from~\cite{krstic2026feedbacklinearizationhyperbolicpdes-trunc}:
\begin{itemize}[leftmargin=1.2em]
\item the local sup-norm Lipschitz constant of the controller kernel (Lemma~5 of~\cite{krstic2026feedbacklinearizationhyperbolicpdes-trunc}),
\begin{equation}
\Lipk(r) := D_K\,e^{\Upsilon_K}\,\frac{C_K r\,(2 - C_K r)}{(1 - C_K r)^2},
\label{eq:LipK-imp}
\end{equation}
defined for $r \in [0, 1/C_K)$;
\item the sup-norm transport propagation lemma for the target system $w_t = w_x$ (Lemma~7 of~\cite{krstic2026feedbacklinearizationhyperbolicpdes-trunc}): $\|w(\cdot,t)\|_\infty \le \max(\|w_0\|_\infty,\, \sup_{0\le\tau\le t}|b(\tau)|)$ for $0 \le t \le 1$, and $\|w(\cdot,t)\|_\infty \le \sup_{t-1 \le \tau \le t}|b(\tau)|$ for $t \ge 1$ (this estimate concerns $w$ alone and is independent of the invertibility of $\mathcal{T}$);
\item the conversion $\|u(\cdot,t)\|_\infty \le \frac{1}{1-\Lipk(r_0)}\,\|w(\cdot,t)\|_\infty$, valid for $\|w(\cdot,t)\|_\infty < \rho_w(r_0) := r_0(1 - \Lipk(r_0))$ (Corollary~1 of~\cite{krstic2026feedbacklinearizationhyperbolicpdes-trunc});
\item the closed-loop well-posedness construction (Lemma~8 of~\cite{krstic2026feedbacklinearizationhyperbolicpdes-trunc}), invoked once in the proof to deliver the unique classical solution.
\end{itemize}

\begin{figure*}[t]
\centering
\resizebox{0.81\linewidth}{!}{%
\begin{tikzpicture}[
    >={Stealth[length=2.5mm]},
    every node/.style={font=\normalsize},
]
\def\xmax{12.5}
\def\ymax{4.7}
\def\rhoZeroY{4.1}
\def\gammaY{1.1}
\def\uZeroY{3.6}
\draw[->, thick] (0,0) -- (\xmax,0) node[below right] {$t$};
\draw[->, thick] (0,0) -- (0,\ymax) node[left] {$\|u(\cdot,t)\|_\infty$};
\draw[densely dashed, thick] (0,\gammaY) -- (\xmax-0.4,\gammaY);
\draw[thick] (-0.12,\rhoZeroY) -- (0.12,\rhoZeroY);
\node[anchor=east] at (-0.18,\rhoZeroY) {$\rho_u^0(r_0)$};
\draw[thick] (-0.12,\uZeroY) -- (0.12,\uZeroY);
\node[anchor=east] at (-0.18,\uZeroY) {$\|u_0\|_\infty$};
\draw[thick] (-0.12,\gammaY) -- (0.12,\gammaY);
\node[anchor=east] at (-0.18,\gammaY) {$\gamma(\varepsilon)$};
\draw[very thick, blue!70!black, smooth] plot coordinates {
    (0.00, 3.60) (0.45, 4.00) (0.90, 3.10) (1.10, 3.55) (1.85, 2.20)
    (2.20, 3.05) (2.55, 2.45) (3.20, 1.55) (3.55, 2.30) (3.85, 1.80)
    (4.45, 1.30) (4.95, 1.85) (5.20, 1.60) (5.85, 1.15) (6.55, 1.55)
    (6.85, 1.25) (7.55, 1.40) (8.10, 1.13) (8.95, 1.30) (9.70, 1.16)
    (11.50, 1.14)
};
\node[align=center, anchor=west, font=\small] at (5.0, 4.0)
    {as the state shrinks, truncation residual $\mathcal{U}_N - \mathcal{U}$,\\
     bounded by $\varepsilon_N(\|u(\cdot,t)\|_\infty)$, shrinks};
\draw[->, thick, gray]
    (4.95, 3.85) to[out=-150, in=70] (2.7, 2.85);
\node[anchor=center, font=\small] at ({\xmax/2 + 0.4}, 0.45)
    {state error persists as neural-operator error $\widehat{\mathcal{U}}_N - \mathcal{U}_N$, bounded by $\varepsilon$, persists};
\draw[decorate, decoration={brace,amplitude=4pt}, thick]
    (\xmax-0.3, \uZeroY) -- (\xmax-0.3, \gammaY);
\node[anchor=west, align=left] at (\xmax-0.2, {(\uZeroY+\gammaY)/2})
    {$\beta(\|u_0\|_\infty, t)$};
\end{tikzpicture}%
}
\caption{Illustration of the practical stability result in part (iv) of Theorem~\ref{thm:stab}, with a linear-in-$\varepsilon$ class-$\mathcal{K}$ residual bound $\gamma(\varepsilon)$. The state norm settles to a residual $\Linf$ ball whose radius is proportional to the neural-operator approximation accuracy $\varepsilon$.}
\label{fig:trajectory}
\end{figure*}
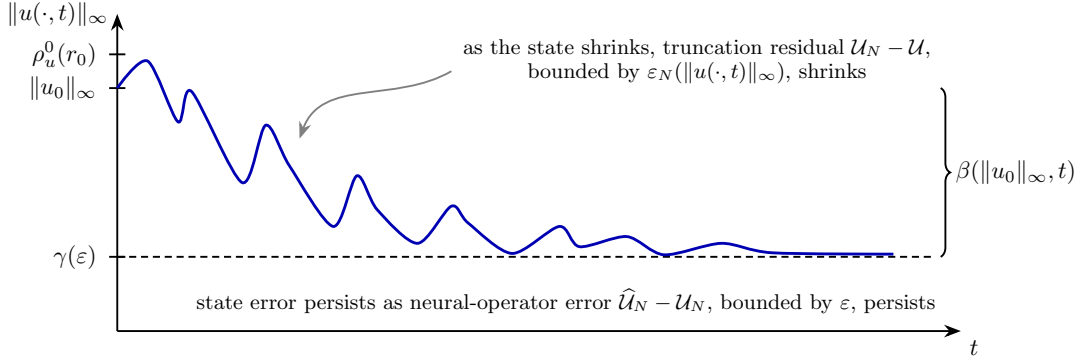

\begin{theorem}[Practical stability]
\label{thm:stab}
Fix $r_0 \in (0, \bar r)$ and $N \ge 2$. Let $\eps > 0$ be the neural-operator accuracy of \eqref{eq:NObnd}, taken on the input class $\mathcal{F}_N^\star(B_f,L_f) \times \mathcal{V}(r_0,L_u)$. Assume the small-gain and basin conditions
\begin{eqnarray}
\Lipk(r_0) &<& 1, \nonumber \\
r_0 &<& \bar r \;:=\; \min(r^*, \rho_f), \nonumber \\
r^* &:=& \Lipk^{-1}(1),
\label{eq:cond_Lipk}
\end{eqnarray}
the local Lipschitz bound on the neural-operator feedback
\begin{eqnarray}
\bigl|\widehat{\mathcal{U}}_N(\mathbf{f}_N, u) - \widehat{\mathcal{U}}_N(\mathbf{f}_N, \tilde u)\bigr| &\le& \hat L_U\,\|u - \tilde u\|_{\Linf(0,1)}, \nonumber \\
&& \forall\, u, \tilde u \in \mathcal{V}(r_0, L_u),
\label{eq:NOlip}
\end{eqnarray}
with $\hat L_U < 1$, the residual condition
\begin{equation}
\eps + \eps_N(r_0) < \rho_w(r_0) := r_0\,(1 - \Lipk(r_0)),
\label{eq:cond_resid_NO}
\end{equation}
and the trajectory-Lipschitz property
\begin{equation}
\Lip\bigl(u(\cdot, t)\bigr) \le L_u, \quad t \ge 0.
\label{eq:traj-Lip}
\end{equation}
Let $u_0 \in \Linf(0,1)$ satisfy
\begin{equation}
\|u_0\|_\infty \;<\; \rho_u^0(r_0) := r_0\,\frac{1 - \Lipk(r_0)}{1 + \Lipk(r_0)}.
\label{eq:cond_init_NO}
\end{equation}
Then the closed-loop system has a unique classical solution $u \in C([0,\infty); \Linf(0,1))$, with the characteristic representation
\begin{eqnarray}
u(x,t) &=& u_0(x+t) \nonumber \\
&& +\, \int_0^t F[u(\cdot,s)](x+t-s)\, ds, \quad x + t \le 1,
\label{eq:char-mild-a}
\end{eqnarray}
\begin{eqnarray}
u(x,t) &=& \widehat{\mathcal{U}}_N\bigl(\mathbf{f}_N, u(\cdot, t-(1-x))\bigr) \nonumber \\
&& +\, \int_{t-(1-x)}^t F[u(\cdot,s)](x+t-s)\, ds, \ x + t > 1,
\label{eq:char-mild-b}
\end{eqnarray}
satisfying the following bounds for every $t \ge 0$.
\begin{itemize}
\item[(i)] \emph{(Forward invariance.)} $\|u(\cdot,t)\|_\infty < r_0$.
\item[(ii)] \emph{(Practical stability.)}
\begin{eqnarray}
\|u(\cdot,t)\|_\infty &\le& \max\biggl(\frac{1+\Lipk(r_0)}{1-\Lipk(r_0)}\,\|u_0\|_\infty, \nonumber \\
&& \qquad\ \, \frac{\eps + \eps_N(r_0)}{1-\Lipk(r_0)}\biggr).
\label{eq:stab-ii}
\end{eqnarray}
\item[(iii)] \emph{(Practical attractivity.)} For all $t \ge 1$,
\begin{equation}
\|u(\cdot,t)\|_\infty \;\le\; \frac{\eps + \eps_N(r_0)}{1-\Lipk(r_0)}.
\label{eq:stab-iii}
\end{equation}
\item[(iv)] \emph{($\eps$-practical asymptotic stability.)} If the initial condition satisfies
\begin{eqnarray}
&& \|u_0\|_\infty \nonumber \\
&& < \rho_u^0(r_0)\biggl(1 - \frac{\eps}{r_0(1-\Lipk(r_0)) - \eps_N(r_0)}\biggr),
\label{eq:cond_init_NO_sharp}
\end{eqnarray}
which is consistent with \eqref{eq:cond_init_NO} at $\eps = 0$ and tightens linearly as $\eps$ grows, then the closed-loop trajectory satisfies
\begin{eqnarray}
\|u(\cdot,t)\|_\infty &\le& \frac{1 + \Lipk(r_0)}{1-\Lipk(r_0)}\,q_0(r_0)^{\max(t-1,\, 0)} \nonumber \\
&& \times\, \|u_0\|_\infty + \gamma(\eps), \quad t \ge 0,
\label{eq:stab-iv}
\end{eqnarray}
where the truncation-residual contraction rate $q_0(r_0)$, monotone increasing in $r_0$, is
\begin{eqnarray}
q_0(r_0) &:=& \frac{\eps_N(r_0)}{r_0\,(1 - \Lipk(r_0))} \nonumber \\
&=& \sup_{0 < s \le r_0}\frac{\eps_N(s)}{s\,(1 - \Lipk(r_0))} \;<\; 1,
\label{eq:q0-def}
\end{eqnarray}
the supremum achieved at the endpoint $s = r_0$ by monotonicity of $\eps_N(s)/s$ on $(0, r_0]$ and the inequality $q_0(r_0) < 1$ following from \eqref{eq:cond_resid_NO}, and the residual envelope is
\begin{equation}
\gamma(\eps) \;:=\; \frac{1}{1 - \Lipk(r_0) - \eps_N(r_0)/r_0}\,\eps.
\label{eq:gamma-def}
\end{equation}
\end{itemize}
\end{theorem}

\begin{remark}
\label{rem:small-gain}
The hypothesis $\hat L_U < 1$ in \eqref{eq:NOlip} is natural: the Volterra feedback $\mathcal{U}_N$ has no linear term (the series begins at $n=2$), so its Lipschitz constant vanishes at the origin and grows continuously with $r_0$, mirroring the basin-induced bound $\Lip_{K_N}(r_0) \le \Lip_K(r_0) < 1$ of~\cite{krstic2026feedbacklinearizationhyperbolicpdes-trunc}. The surrogate inherits this small-gain behavior when trained on data from a small basin where $\mathcal{U}_N$ has small slope, optionally with Lipschitz-constrained layers.
\end{remark}

\begin{remark}
\label{rmk:residual-eps-N}
The residual $\eps + \eps_N(r_0)$ in (ii)--(iii) has two independent origins: the network error $\eps$, controlled by training, and the truncation tail $\eps_N(r_0)$, controlled by the order $N$. The envelope $\gamma(\eps)$ in (iv) scales linearly with $\eps$, with $N$ entering only through $\eps_N(r_0)$; as $\eps \to 0$ the $\mathcal{KL}$ stability of~\cite{krstic2026feedbacklinearizationhyperbolicpdes-trunc} is recovered.
\end{remark}

\section{Proof of Theorem~\ref{thm:stab}}
\label{sec:proof-stab}

\emph{Local well-posedness.} By \eqref{eq:cond_init_NO} and $\Lipk(r_0)<1$, $\|u_0\|_\infty < \rho_u^0(r_0) < r_0$. The construction of Lemma~8 of~\cite{krstic2026feedbacklinearizationhyperbolicpdes-trunc} applies with the truncated boundary feedback $\mathcal{U}_N(\mathbf{f}_N, u(\cdot,t))$ replaced by $\widehat{\mathcal{U}}_N(\mathbf{f}_N, u(\cdot,t))$ and with the kernel-feedback Lipschitz constant $\Lip_{K_N}(r_0)$ replaced by $\hat L_U$. The contraction estimate on a time interval of length $\tau$ becomes $\hat L_U + \tau\, L_F(r_0) < 1$, which is met for $\tau < (1 - \hat L_U)/L_F(r_0)$ since $\hat L_U < 1$ by \eqref{eq:NOlip}, yielding the Banach fixed point on $C([0, \tau]; \Linf(0,1))$. Concatenation of intervals furnishes a unique classical solution on a maximal interval $[0, T_{\max})$ satisfying \eqref{eq:char-mild-a}--\eqref{eq:char-mild-b}, with the blow-up alternative
\begin{eqnarray}
T_{\max} < \infty &\Rightarrow& \limsup_{t \uparrow T_{\max}}\|u(\cdot,t)\|_\infty \ge r_0.
\label{eq:blowup-alt}
\end{eqnarray}

\emph{A priori residual bound on $[0, T_{\max})$.} On $[0, T_{\max})$ the solution satisfies $\|u(\cdot,t)\|_\infty < r_0$ by definition of $T_{\max}$ together with the continuation criterion. By \eqref{eq:NObnd} the neural-operator error in the first bracket of \eqref{eq:bdef-NO} is bounded by $\eps$, and by~\cite{krstic2026feedbacklinearizationhyperbolicpdes-trunc} the truncation tail in the second bracket is bounded by $\eps_N(\|u(\cdot,t)\|_\infty) \le \eps_N(r_0)$. Thus $|b(t)| \le \eps + \eps_N(r_0)$ for $t \in [0, T_{\max})$.

\emph{(i) Forward invariance, with closure $T_{\max} = \infty$.} The transport propagation lemma (Lemma~7 of~\cite{krstic2026feedbacklinearizationhyperbolicpdes-trunc}) gives $\|w(\cdot,t)\|_\infty \le \max(\|w_0\|_\infty,\, \eps + \eps_N(r_0))$. By \eqref{eq:cond_init_NO} and the algebraic identity $\rho_w(r_0) = r_0(1 - \Lipk(r_0))$,
\begin{equation}
\|w_0\|_\infty \;\le\; (1 + \Lipk(r_0))\,\|u_0\|_\infty \;<\; \rho_w(r_0),
\label{eq:w0-bnd-stab}
\end{equation}
and by \eqref{eq:cond_resid_NO}, $\eps + \eps_N(r_0) < \rho_w(r_0)$. Hence $\|w(\cdot,t)\|_\infty < \rho_w(r_0)$, so the conversion $\|u(\cdot,t)\|_\infty \le \frac{1}{1-\Lipk(r_0)}\,\|w(\cdot,t)\|_\infty$ (Corollary~1 of~\cite{krstic2026feedbacklinearizationhyperbolicpdes-trunc}) applies and yields the strict bound $\|u(\cdot,t)\|_\infty < \frac{1}{1-\Lipk(r_0)}\,\rho_w(r_0) = r_0$ for $t \in [0, T_{\max})$. By the contrapositive of the blow-up alternative \eqref{eq:blowup-alt}, $T_{\max} = \infty$, and (i) holds for every $t \ge 0$.

\emph{(ii) Practical stability.} Combining the conversion with the transport propagation,
\begin{eqnarray}
\|u(\cdot,t)\|_\infty
&\le& \frac{1}{1-\Lipk(r_0)}\,\max\bigl(\|w_0\|_\infty, \nonumber \\
&& \qquad\qquad\qquad\ \, \eps + \eps_N(r_0)\bigr) \nonumber \\
&\le& \max\Bigl(\frac{1+\Lipk(r_0)}{1-\Lipk(r_0)}\,\|u_0\|_\infty, \nonumber \\
&& \qquad\ \frac{\eps + \eps_N(r_0)}{1-\Lipk(r_0)}\Bigr),
\label{eq:ii-bnd}
\end{eqnarray}
using \eqref{eq:w0-bnd-stab} for the first argument of the maximum.

\emph{(iii) Practical attractivity.} For $t \ge 1$, the transport lemma (Lemma~7 of~\cite{krstic2026feedbacklinearizationhyperbolicpdes-trunc}) gives $\|w(\cdot,t)\|_\infty \le \sup_{t-1\le\tau\le t}|b(\tau)| \le \eps + \eps_N(r_0)$, and the conversion gives \eqref{eq:stab-iii}.

\emph{(iv) $\eps$-practical asymptotic stability.} Let $Y_k := \sup_{t \ge k}\|u(\cdot,t)\|_\infty$ and $\alpha(s) := \frac{1+\Lipk(r_0)}{1-\Lipk(r_0)}\,s$. The function $s \mapsto \eps_N(s)/s$ is increasing on $(0, r_0]$ (since $\eps_N(s)/s = (D_K e^{\Upsilon_K}/C_K)\,C_K^{N+1}\,s^N/(1 - C_K s)$ is the product of two increasing positive functions of $s$), so the supremum in \eqref{eq:q0-def} is attained at $s = r_0$, and
\begin{equation}
\frac{1}{1-\Lipk(r_0)}\,\eps_N(s) \;\le\; q_0(r_0)\, s, \quad 0 \le s \le r_0.
\label{eq:epsN-comp}
\end{equation}

\emph{Tail-supremum recursion.} For $t \ge k+1$, Lemma~7 of~\cite{krstic2026feedbacklinearizationhyperbolicpdes-trunc} applied on the window $[t-1, t] \subseteq [k, \infty)$ together with monotonicity of $\eps_N$ gives $\|w(\cdot,t)\|_\infty \le \eps + \eps_N(Y_k)$, and the conversion (which applies since $\eps + \eps_N(Y_k) \le \eps + \eps_N(r_0) < \rho_w(r_0)$) combined with \eqref{eq:epsN-comp} yields
\begin{eqnarray}
Y_{k+1} &\le& \frac{\eps + \eps_N(Y_k)}{1-\Lipk(r_0)} \nonumber \\
&\le& q_0(r_0)\, Y_k + \frac{\eps}{1-\Lipk(r_0)}, \quad k \ge 0.
\label{eq:Yk-rec}
\end{eqnarray}
Using the identity $\frac{\eps}{1-\Lipk(r_0)} = (1 - q_0(r_0))\,\gamma(\eps)$, equivalent to \eqref{eq:gamma-def}, this rewrites as
\begin{eqnarray}
Y_{k+1} - \gamma(\eps)
&\le& q_0(r_0)\,\bigl(Y_k - \gamma(\eps)\bigr), \quad k \ge 0.
\label{eq:Yk-rec-shifted}
\end{eqnarray}

\emph{Excess sequence.} Define $Z_k := \max(Y_k - \gamma(\eps),\, 0)$. From \eqref{eq:Yk-rec-shifted}, $Y_{k+1} - \gamma(\eps) \le q_0(r_0)\,(Y_k - \gamma(\eps))$, hence
\begin{eqnarray}
Z_{k+1} &=& \max\bigl(Y_{k+1} - \gamma(\eps),\, 0\bigr) \nonumber \\
&\le& \max\bigl(q_0(r_0)(Y_k - \gamma(\eps)),\, 0\bigr) \nonumber \\
&=& q_0(r_0)\,Z_k, \quad k \ge 0.
\label{eq:Zk-rec}
\end{eqnarray}

\emph{Initial bound on the excess.} The transport propagation lemma has two regimes: for $0 \le t \le 1$, $\|w(\cdot,t)\|_\infty \le \max(\|w_0\|_\infty,\, \sup_{0\le\tau\le t}|b(\tau)|)$; for $t \ge 1$, $\|w(\cdot,t)\|_\infty \le \sup_{t-1\le\tau\le t}|b(\tau)|$. In either regime, $\|u(\cdot,\tau)\|_\infty \le Y_0$ together with monotonicity of $\eps_N$ gives $|b(\tau)| \le \eps + \eps_N(Y_0)$ for every $\tau \ge 0$. Combining the two regimes, $\|w(\cdot,t)\|_\infty \le \max(\|w_0\|_\infty,\, \eps + \eps_N(Y_0))$ for $t \ge 0$. Both arguments of the maximum are strictly less than $\rho_w(r_0)$ (the first by \eqref{eq:w0-bnd-stab}, the second by \eqref{eq:cond_resid_NO} and $\eps_N(Y_0) \le \eps_N(r_0)$ via $Y_0 \le r_0$ from (i)), so the conversion applies and yields, on taking the supremum over $t \ge 0$ and using $\eps_N(Y_0) \le q_0(r_0)\,Y_0$ from \eqref{eq:epsN-comp} together with the identity $\frac{\eps}{1-\Lipk(r_0)} = (1-q_0(r_0))\,\gamma(\eps)$ from \eqref{eq:gamma-def},
\begin{eqnarray}
Y_0 &\le& \max\bigl(\alpha(\|u_0\|_\infty), \nonumber \\
&& \quad q_0(r_0)\,Y_0 + (1-q_0(r_0))\,\gamma(\eps)\bigr).
\label{eq:Y0-amp}
\end{eqnarray}
If $Y_0 \le \alpha(\|u_0\|_\infty)$, then $Z_0 = \max(Y_0 - \gamma(\eps), 0) \le \alpha(\|u_0\|_\infty)$ directly. Otherwise $Y_0 > \alpha(\|u_0\|_\infty)$, in which case the maximum in \eqref{eq:Y0-amp} is necessarily attained by the second argument, so $Y_0 \le q_0(r_0)\,Y_0 + (1-q_0(r_0))\,\gamma(\eps)$. Solving for $Y_0$ (using $q_0(r_0) < 1$),
\begin{eqnarray}
(1-q_0(r_0))\,Y_0 &\le& (1-q_0(r_0))\,\gamma(\eps), \nonumber \\
\text{i.e.,} \quad Y_0 &\le& \gamma(\eps),
\label{eq:Y0-case2}
\end{eqnarray}
hence $Z_0 = 0 \le \alpha(\|u_0\|_\infty)$. In both cases,
\begin{equation}
Z_0 \;\le\; \alpha(\|u_0\|_\infty).
\label{eq:Z0-bound}
\end{equation}

\emph{Class-$\mathcal{KL}$ envelope.} Iterating \eqref{eq:Zk-rec} from \eqref{eq:Z0-bound}, $Z_k \le q_0(r_0)^k\,\alpha(\|u_0\|_\infty)$ for $k \ge 0$, and hence $Y_k \le q_0(r_0)^k\,\alpha(\|u_0\|_\infty) + \gamma(\eps)$. For $t \ge 1$,
\begin{eqnarray}
\|u(\cdot,t)\|_\infty &\le& Y_{\lfloor t \rfloor} \nonumber \\
&\le& q_0(r_0)^{\lfloor t \rfloor}\alpha(\|u_0\|_\infty) + \gamma(\eps) \nonumber \\
&\le& q_0(r_0)^{t-1}\alpha(\|u_0\|_\infty) + \gamma(\eps).
\label{eq:u-bound-tge1}
\end{eqnarray}
For $t \in [0, 1]$,
\begin{eqnarray}
\|u(\cdot,t)\|_\infty &\le& Y_0 \nonumber \\
&\le& \alpha(\|u_0\|_\infty) + \gamma(\eps) \nonumber \\
&=& q_0(r_0)^0\,\alpha(\|u_0\|_\infty) + \gamma(\eps).
\label{eq:u-bound-tle1}
\end{eqnarray}
Combining, $\|u(\cdot,t)\|_\infty \le \alpha(\|u_0\|_\infty)\,q_0(r_0)^{\max(t-1, 0)} + \gamma(\eps)$, which is \eqref{eq:stab-iv}. The envelope $\beta(s, t) := \alpha(s)\,q_0(r_0)^{\max(t-1, 0)}$ is class-$\mathcal{KL}$: linear in $s$ with $\beta(0, t) = 0$, and geometrically decaying in $t$ at rate $q_0(r_0) < 1$.

\section{Numerical Results}
\label{sec:numerics}
\begin{figure*}[t]
  \centering
  \includegraphics[width=\linewidth]{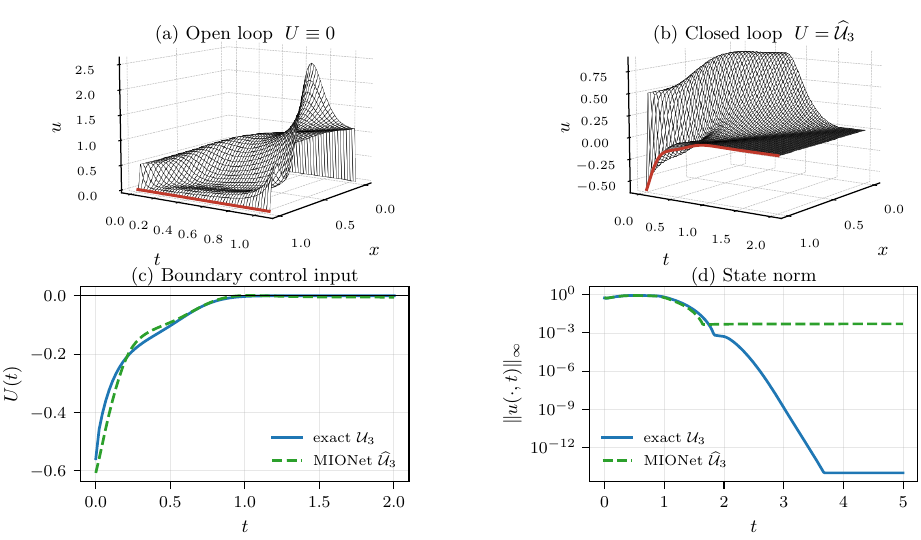}
  \caption{Closed-loop comparison at $\bm\gamma = (3,3,3,3)$, $A=4$, Initial condition has a max amplitude of $\|u_0\|_\infty = 0.5$.
    \textbf{(a)} Open-loop instability of the state $u_{\rm open}(x,t)$ with $U \equiv 0$.
    \textbf{(b)} Closed-loop state $u(x,t)$ under $\widehat{\mathcal{U}}_3$; the red curve overlays the controlled-boundary trace $u(1,t)=U(t)$.
    \textbf{(c)} Boundary control $U(t)$ for the exact $\mathcal{U}_3$ (solid blue) and the MIONet $\widehat{\mathcal{U}}_3$ (dashed green).
    \textbf{(d)} State norm $\|u(\cdot, t)\|_\infty$. Exact $\mathcal{U}_3$ reaches numerical zero; MIONet plateaus at $\sim 5\!\times\!10^{-3}$.}
  \label{fig:closed-loop}
\end{figure*}
\begin{figure*}[t]
  \centering
  \includegraphics[width=0.85\linewidth]{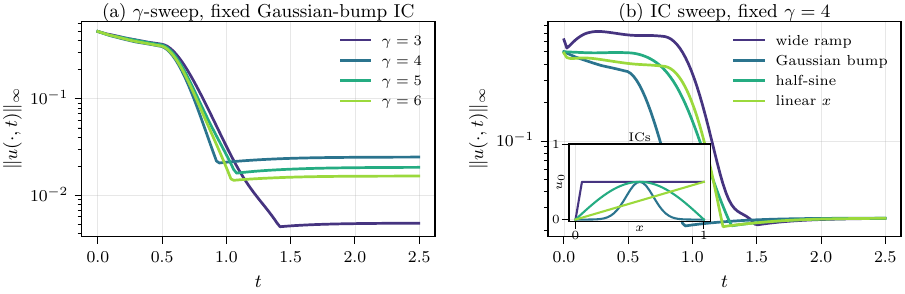}
  \caption{Robustness of $\widehat{\mathcal{U}}_3$ (no retraining), $\|u(\cdot, t)\|_\infty$ on a log axis.
    \textbf{(a)} $\bm\gamma$-sweep at fixed Gaussian IC, $A=4$: $\gamma \in \{3, 4, 5, 6\}$.
    \textbf{(b)} IC-sweep at fixed $\gamma = A = 4$; inset overlays the four $u_0(x)$.}
  \label{fig:sweep}
\end{figure*}
\begin{figure*}[t]
  \centering
  \includegraphics[width=0.85\linewidth]{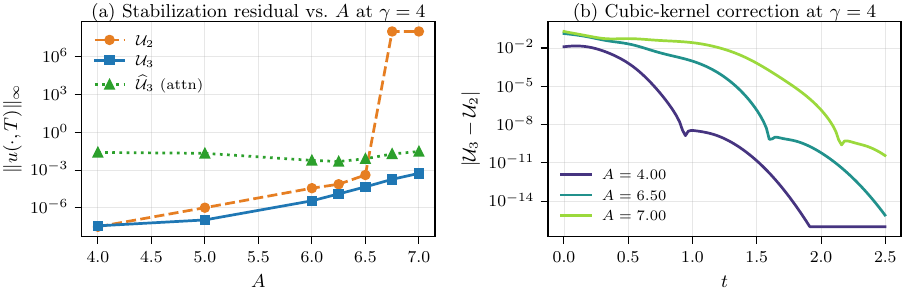}
  \caption{Truncation-order ablation at fixed $\gamma = 4$, Gaussian IC.
    \textbf{(a)} Final residual $\|u(\cdot, T)\|_\infty$ vs.\ $A$; $\mathcal{U}_2$ blows up at $A \ge 6.75$, $\mathcal{U}_3$ and $\widehat{\mathcal{U}}_3$ hold.
    \textbf{(b)} Correction $|\mathcal{U}_3(t) - \mathcal{U}_2(t)|$ for $A \in \{4, 6.5, 7\}$.}
  \label{fig:trunc}
\end{figure*}
We demonstrate the framework by learning the boundary controller directly as the operator $(f_2, f_3, u) \longmapsto \mathcal{U}_3$ at truncation order $N = 3$, where $\mathcal{U}_3$ is the kernel-based feedback of \eqref{eq:KNdef}; the order-$3$ kernels are constructed in closed form in~\cite{krstic2026feedbacklinearizationhyperbolicpdes} (see Appendix~\ref{subsec:plant-controller}).

On a common training corpus (described below) we run three closed-loop experiments and a runtime benchmark: closed-loop stabilization on a representative plant, robustness across plant parameters and initial conditions without retraining, a truncation-order ablation, and a wall-clock comparison of $\mathcal{U}_2$, $\mathcal{U}_3$, and $\widehat{\mathcal{U}}_3$. The surrogate architecture and training are in Appendix~\ref{appendix:neural-operator-design}.

\vspace{-1.5\baselineskip}
\paragraph*{\normalfont\bfseries Dataset construction.}
\label{subsec:data}
The plant coefficients are tensor products of a squared Chebyshev factor $\varphi(\gamma, z) := \cos^2(\gamma \arccos z)$, namely $f_2 = A^3\,\varphi(\gamma_2^x, x)\prod_{i}\varphi(\gamma_2^\xi, \xi_i)$ and $f_3 = A^4\,\varphi(\gamma_3^x, x)\prod_i \varphi(\gamma_3^\xi, \xi_i)$, with shape scalars $\bm\gamma := (\gamma_2^x, \gamma_2^\xi, \gamma_3^x, \gamma_3^\xi) \sim \mathcal{U}[3,6]^4$ and amplitude $A \sim \mathcal{U}[4,7]$ drawn per plant; the squared form keeps the Volterra forcing signed, and since $|\varphi|\le1$ the cube gives $\|f_n\|_\infty \le A^{n+1}$, $\operatorname{Lip}(\varphi) \le 2\gamma^2$, a band-limited subclass of $\mathcal{F}_3^\star(B_f, L_f)$~\cite{Bhan2024Neural}. States are drawn from a Mat\'ern-$3/2$ Gaussian field (length scale $\ell = 0.15$, $\sigma = 0.4$), shifted to $u(0)=0$ and rejection-sampled to $\|u\|_\infty \le R = 1$, $\operatorname{Lip}(u) \le L_u = 25$; the bounds and $\ell$ are calibrated to a pilot $\mathcal{U}_3$ closed loop, so the corpus covers states the simulator actually visits. Sampling $200$ plants $\times\, 500$ states gives $M = 100{,}000$ triples $(\bm\gamma, A, u, \mathcal{U}_3)$; kernels $k_2, k_3$ are built once per plant on a simplex grid ($N_x = 40$) and targets evaluated by Simpson quadrature (Appendix~\ref{appendix:numerical-solver-details}).

\vspace{-1.5\baselineskip}
\paragraph*{\normalfont\bfseries Neural-operator experiments.}
\label{subsec:training}

The surrogate $\widehat{\mathcal{U}}_3$ is a small transformer-encoder network taking the three functions $(f_2, f_3, u)$ on the grid --- not the parameters $(\bm\gamma, A)$ --- so it generalizes to any admissible plant with a grid representation. It is trained on the $100{,}000$-triple corpus, split $80/10/10$ by plant, with AdamW~\cite{loshchilov2018decoupled} for $200$ epochs; architecture~\cite{Vaswani2017Attention,Hao2023GNOT,Jin2022MIONet}, hyperparameters, and diagnostics are in Appendix~\ref{appendix:neural-operator-design}. The deployed operator reaches $30.6\%$ relative RMSE over $20$ unseen $(\bm\gamma, A)$ plants.

\vspace{-1.5\baselineskip}
\paragraph*{Closed-loop example.}
Figure~\ref{fig:closed-loop} shows a representative closed loop. Open-loop ($U \equiv 0$) the plant blows up; under both $\mathcal{U}_3$ and $\widehat{\mathcal{U}}_3$ it stabilizes, the exact $\mathcal{U}_3$ to machine precision and the surrogate to a residual floor of order $10^{-3}$ set by its approximation error --- the practical stability of Theorem~\ref{thm:stab}. Without retraining, $\widehat{\mathcal{U}}_3$ stabilizes across a $\gamma$-sweep ($\gamma \in \{3,4,5,6\}$, fixed $A=4$) and an initial-condition sweep (ramp, Gaussian, half-sine, linear; fixed $\gamma = A = 4$), all settling near $10^{-2}$ (Figure~\ref{fig:sweep}). The truncation order matters at large coefficient Lipschitz: sweeping $A \in \{4,\ldots,7\}$ at $\bm\gamma = (4,4,4,4)$ (Figure~\ref{fig:trunc}), the quadratic-only $\mathcal{U}_2$ fails at $A \ge 6.75$ while $\mathcal{U}_3$ and $\widehat{\mathcal{U}}_3$ hold, though the controller difference $|\mathcal{U}_3 - \mathcal{U}_2|$ is only of order $10^{-2}$.

\begin{table}[t]
  \centering
  \caption{Per-mapping runtime on CUDA over $5$ seeds, $N = 100$ random plants per seed.}
  \label{tab:runtime}
  \small
  \begin{tabular}{@{}lrr@{}}
    \toprule
    Method & Mean $\pm$ std (ms) & Speedup \\
    \midrule
    $\mathcal{U}_2$ ($k_2$)      &    $0.63 \pm 0.06$  & $6296\times$ \\
    $\mathcal{U}_3$ ($k_2,k_3$)  & $3947.29 \pm 15.64$ & $1\times$ \\
    MIONet $\widehat{\mathcal{U}}_3$ &    $0.998 \pm 0.076$ & $3957\times$ \\
    \bottomrule
  \end{tabular}
\end{table}
The payoff of the surrogate is amortized online cost: the classical $\mathcal{U}_3$ pipeline rebuilds $k_2, k_3$ and runs nested Simpson quadrature at $\approx 3.95$~s per call, while the trained $\widehat{\mathcal{U}}_3$ evaluates in $\approx 1.0$~ms --- a $\mathbf{\sim\!4000\times}$ speedup, at a one-time training cost. Table~\ref{tab:runtime} reports per-mapping wall-clock over $100$ random $(\bm\gamma, A, u)$ triples per seed across $5$ seeds, all on a single NVIDIA RTX 4070 SUPER, with both classical pipelines implemented in PyTorch on the same GPU.

\section{Conclusion}
\label{sec:conclusion}

We replaced the plant-specific truncated Volterra feedback --- which recomputes $N-1$ kernels and an $N$-fold nested integration at every control update --- with a single neural-operator surrogate over plants and states, evaluated in constant time. The surrogate is admissible because the truncated map is jointly Lipschitz on a compact class, and the closed loop tolerates both the truncation residual $\eps_N(r_0)$ and the network error $\eps$ at once. All four stability properties of Part~I survive, with a class-$\mathcal{KL}$ envelope in the initial condition and linear dependence on $\eps$, recovering the exact design as $\eps \to 0$.

Numerically, the learned operator evaluates the feedback in $1$~millisecond, against $3.9$~seconds to assemble the kernels and integrate --- a speedup of $\approx 4000\times$  (Table~\ref{tab:runtime}). The neural operator is trained once and stabilizes plants unseen in training, including where the cubic kernel is needed and $\mathcal{U}_2$ fails.

\appendix
\section{Numerical Details}
\label{app:numerical-details}
\label{subsec:plant-controller}
\label{appendix:numerical-solver-details}
\label{appendix:neural-operator-design}

At $N=3$ the plant carries only quadratic and cubic coefficients $f_2, f_3$, and the kernels $k_2, k_3$ are the order-$3$ instance of the characteristic recursion of~\cite{krstic2026feedbacklinearizationhyperbolicpdes} (equation~\eqref{eq:kndef}), where the explicit closed forms are derived. No kernel PDE is solved: each $k_n$ reduces to evaluations of $f_2, f_3$ at characteristic-shifted arguments and quadrature on intervals, and the controller is the boundary trace $\mathcal{U}_{3,k}(\mathbf{k}_3, u) = \int_{T_2(1)} k_2(1,\bm\xi)\, u^{\otimes 2}\, d\bm\xi + \int_{T_3(1)} k_3(1,\bm\xi)\, u^{\otimes 3}\, d\bm\xi$.

The simulator uses method-of-lines: first-order upwind~\cite{leveque2007finite} for $u_x$ on a uniform grid ($N_x = 40$, $\Delta x = 0.025$) and explicit RK4 in time with $\Delta t = 0.9\,\Delta x$ (CFL ratio $0.9$); integrals use Simpson's rule~\cite{burden2015numerical}.

The surrogate's input is the three functions $f_2, f_3, u$ sampled on the simulator grid ($N_x = 40$), so it sees the same on-grid values a finite-difference solver would and the $(\bm\gamma, A)$ parameterization is never disclosed --- any plant family can be used without redesign. The network follows MIONet~\cite{Jin2022MIONet} (the multi-input lifting of DeepONet~\cite{Lu2021Universal}) with an attention encoder~\cite{Cao2021Galerkin,Hao2023GNOT,Vaswani2017Attention}: the three channels are tokenized on the grid with learnable positional and channel embeddings, a prepended learnable query row is mixed through a two-layer, four-head pre-norm Transformer encoder (embedding dimension $64$, GELU MLPs of width $256$, $107{,}585$ parameters), and its output row is read out through a two-layer MLP to the scalar $\widehat{\mathcal{U}}_3$. Training is standardized MSE by AdamW for $200$ epochs (cosine schedule, dropout $0.1$), split $80/10/10$ by plant so validation and test plants are unseen $(\bm\gamma, A)$ pairs; the final test relative RMSE is $30.6\%$.

\vspace*{-1em}
\paragraph*{\normalfont\bfseries Acknowledgment.}
The author thanks Luke Bhan, who declined coauthorship of this primarily theoretical paper. Luke developed the numerical results in Section \ref{sec:numerics} and the Appendix.

\bibliographystyle{plain}
\bibliography{bib-hyp-FL}

\end{document}